\documentclass[11pt]{article}

\usepackage{indentfirst}
\usepackage{amsmath}
\usepackage{amsthm}
\usepackage[dvipsnames]{xcolor}
\usepackage{eufrak}
\usepackage{verbatim}
\usepackage[makeroom]{cancel}
\usepackage{rotating}
\usepackage{float}
\usepackage{fullpage}
\usepackage{amssymb}
\usepackage{mathrsfs}
\usepackage{epsfig}
\usepackage{graphicx}
\usepackage{framed}
\usepackage[numbers]{natbib}
\usepackage{color}
\usepackage{hyperref}
\usepackage{bbm}
\usepackage{cleveref}
\usepackage{mathtools}
\usepackage{./shortcuts}
\usepackage{./shortcuts_alphabet}

\usepackage{nth}
\usepackage{amsmath}
\usepackage{microtype}
\usepackage{datetime}
\usepackage{wrapfig}
\usepackage{caption}
\usepackage{ifpdf}
\usepackage[all,arc,curve,color,frame]{xy}
\usepackage{color}
\usepackage{hyperref}
\usepackage{pgf}
\usepackage{tikz}

\usepackage[labelfont=bf]{caption}

\definecolor{corlinks}{RGB}{0,0,150}
\definecolor{cormenu}{RGB}{0,0,150}
\definecolor{corurl}{RGB}{0,0,150}

\hypersetup{
	colorlinks=true,
	urlcolor=corlinks,
	linkcolor=corlinks,
	menucolor=cormenu,
	citecolor=corlinks,
	pdfborder= 0 0 0
}

\newtheorem{theorem}{Theorem}

\newtheorem{lemma}[theorem]{Lemma}
\newtheorem{corollary}[theorem]{Corollary}
\newtheorem{definition}[theorem]{Definition}
\newtheorem{proposition}[theorem]{Proposition}

\DeclareMathOperator{\poly}{poly}
\DeclareMathOperator{\polylog}{polylog}

\DeclareMathOperator*\Exp{{\bf E}}
\DeclareMathOperator*\Prob{{\bf Pr}}
\newcommand{\bool}{\left\{0,1\right\}}

\newcommand{\bnote}[1]{\noindent \textcolor{BrickRed}{(\textbf{Bruno:} #1\noindent)}}

\newcommand{\znote}[1]{\noindent \textcolor{blue}{(\textbf{Zhenjian:} #1\noindent)}}
\renewcommand{\bnote}[1]{}
\renewcommand{\znote}[1]{}

\def\colorful{1}

\ifnum\colorful=1

\fi
\ifnum\colorful=0

\fi

\usepackage{algorithm}
\usepackage[noend]{algpseudocode}

\usepackage{titlesec}

\titleformat{\paragraph}
{\normalfont\normalsize\bfseries}{\theparagraph}{1em}{}
\titlespacing*{\paragraph}
{0pt}{3.25ex plus 1ex minus .2ex}{1.5ex plus .2ex}

\begin{document}
	
	\title{
		Algorithms and Lower Bounds for Comparator Circuits \\from Shrinkage
	}
	
	\author{
		Bruno P. Cavalar\footnote{Email: \texttt{Bruno.Pasqualotto-Cavalar@warwick.ac.uk}}\vspace{0.2cm}\\{\small Department of Computer Science}\\{\small University of Warwick\vspace{0.3cm}}
		\and	
		Zhenjian Lu\footnote{Email: \texttt{Zhen.J.Lu@warwick.ac.uk}}\vspace{0.2cm}\\{\small Department of Computer Science}\\
		{\small University of Warwick} 
	}
	
	\date{}  
	
	\maketitle
	
	\vspace{-0.45cm}
	
	\begin{abstract}
		Comparator circuits are a natural circuit model for studying bounded fan-out computation whose
		power sits between nondeterministic branching programs and general circuits. Despite having been studied for nearly three decades, the first
		superlinear lower bound against comparator circuits was proved only recently by Gál and Robere (ITCS 2020), who established a $\Omega\!\left((n/\log n)^{1.5}\right)$ lower bound on the size of comparator circuits computing an explicit function of $n$ bits. 
		
		In this paper, we initiate the study of \emph{average-case complexity} and \emph{circuit analysis algorithms} for comparator circuits. Departing from previous approaches, we exploit the technique of shrinkage under random restrictions to obtain a variety of new results for this model. Among them, we show
		\begin{itemize}
			\vspace{0.2cm}
			\item \textbf{Average-case Lower Bounds.} 
			For every $k = k(n)$ with $k \geq \log n$, there exists a polynomial-time computable function $f_k$ on $n$ bits such that, for every comparator circuit $C$ with at most
			$n^{1.5}/O\!\left(k\cdot \sqrt{\log n}\right)$ gates, we have
			\[
				\Prob_{x\in\bool^n}\left[C(x)=f_k(x)\right]\leq \frac{1}{2} + \frac{1}{2^{\Omega(k)}}.  
			\]
			This average-case lower bound matches the worst-case lower bound of G{\'{a}}l and Robere by letting $k=O\!\left(\log n\right)$.
			\vspace{0.2cm}
			\item \textbf{$\#$SAT Algorithms.} There is an algorithm that counts the number of satisfying assignments of a given comparator circuit with at most $n^{1.5}/O\!\left(k\cdot \sqrt{\log n}\right)$ gates, in time $2^{n-k}\cdot\poly(n)$, for any $k\leq n/4$. The running time is non-trivial	(i.e., $2^n/n^{\omega(1)}$) when $k=\omega(\log n)$.
			\vspace{0.2cm}
			\item \textbf{Pseudorandom Generators and ${\sf MCSP}$ Lower Bounds.} There is a pseudorandom generator of seed length $s^{2/3+o(1)}$ that fools comparator circuits with $s$ gates. Also, using this PRG, we obtain an $n^{1.5-o(1)}$ lower bound for ${\sf MCSP}$	against comparator circuits.
		\end{itemize}
	\end{abstract}

	\newpage
	\setcounter{tocdepth}{2}
	
	\tableofcontents
	
	\section{Introduction}
	
	A comparator circuit is a Boolean circuit whose gates are
	\emph{comparator gates}, each of which maps a pair of inputs
	$(x,y)$ to $(x \land y,\, x \lor y)$, and whose inputs are labelled
	by a literal (i.e., a variable $x_i$ or its negation $\neg x_i$).
	A convenient way of representing a comparator circuit
	is seen in Figure~\ref{fig:f0}.
	We draw a set of
	horizontal lines,
	each of which is called a \emph{wire} and is labelled by an input literal.
	The gates are represented as a sequence of vertical
	arrows, each of which connects some wire to another. The
	tip of the arrow is the logical disjunction gate~($\lor$),
	and the rear of the arrow is the logical conjunction gate~($\land$). 
	One of the wires is selected 
	to represent the Boolean value
	of the computation.
	The \emph{size} of the circuit is the number of gates in the circuit.
	
	\begin{figure}[h]
        \begin{center}
            \includegraphics{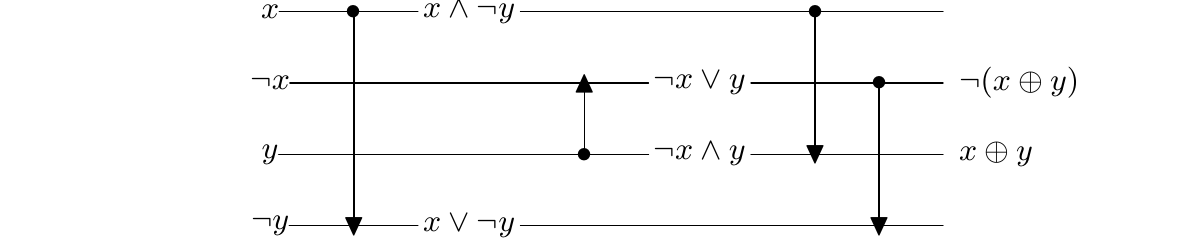}
        \end{center}
		\caption{A comparator circuit with 2 inputs, 4 wires and 4 gates. The third
			wire computes the parity of the two bits. 
		}
		\label{fig:f0}
	\end{figure}
	
	Comparator circuits can be viewed as a restricted type of circuit in
	which the gates have bounded fan-out.
	It is
	easy to see that comparator circuits can efficiently simulate Boolean
    formulas over $\{\land,\lor,\lnot\}$ with no
    overhead\footnote{As a comparison, note that there are linear-size comparator circuits
    for $\parity$~(see \Cref{fig:f0}), whereas any Boolean formula computing $\parity$ has size
    $\Omega(n^2)$~\cite{khrapchenko_1971}.}.
	Moreover, it is also
	known that polynomial-size comparator circuits can even simulate
	nondeterministic branching programs
	\cite{DBLP:journals/jcss/MayrS92}
	with a polynomial overhead only. 
    On the other hand, comparator
    circuits appear to be much stronger than formulas\footnote{Recall that
    the class of polynomial-size formulas is exactly ${\sf NC}^1$.}, 
    as it is conjectured that polynomial-size comparator circuits are
    incomparable to ${\sf NC}$~\cite{DBLP:journals/jcss/MayrS92}.
    Evidence for this conjecture is that polynomial-size comparator
    circuits can compute problems whose known algorithms are inherently
    sequential, such as stable marriage and lexicographically first maximal
    matching
	\cite{DBLP:journals/jcss/MayrS92}, and there is an oracle separation
	between ${\sf NC}$ and 
	polynomial-size comparator circuits
	\cite{DBLP:journals/toct/CookFL14}. 
	Moreover, Robere, Pitassi, Rossman and Cook~\cite{DBLP:conf/focs/RoberePRC16}
	showed that there exists a Boolean function in $\mNC^2$ 
	not computed by 
    polynomial-size
	\emph{monotone} comparator circuits\footnote{
		Comparator circuits are monotone if they don't have negated
		literals.}. 
	For these reasons, comparator circuits are likely to be incomparable to
	$\NC$,
	and polynomial-size span programs,
	which are contained in $\NC^2$, are not expected to be stronger
	than polynomial-size comparator circuits.
	
	Despite the importance of comparator circuits,
	we don't know much about them.
	Though it is easy to see that $\parity$ can be computed by comparator
	circuits with $O(n)$ wires and gates (See Figure~\ref{fig:f0}),
	the best known comparator circuit for $\maj$ uses $O(n)$ wires and
	$O(n \log n)$ gates~\cite{aks_sorting}. 
	We don't know if there is a linear-size comparator circuit for
	$\maj$\footnote{As opposed to a sorting network,
	note that a comparator circuit can use multiple copies of the same input literal.},
	whereas, for the weaker model of nondeterministic branching programs,
	superlinear lower bounds are known~\cite{Raz90b}.
    Structural questions about comparator circuits have also received some
    attention in recent
    years~\cite{DBLP:conf/innovations/GoosKRS19,DBLP:journals/iandc/KomarathSS18}.

	The first
	superlinear
	\emph{worst-case lower bound} for comparator
	circuits
	was recently obtained by
	G{\'{a}}l and Robere \cite{DBLP:conf/innovations/GalR20},
	by an adaptation of Ne\v{c}iporuk's argument~\cite{nechiporuk_1966}.
	Their proof yields a lower bound of
	$\Omega\!\left((n/\log n)^{1.5}\right)$
	to comparator circuits computing a function of $n$ bits.
	For \emph{monotone} comparator circuits,
	exponential lower bounds are known~\cite{DBLP:conf/focs/RoberePRC16}.
	
	In this paper, we exploit structural properties of small-size comparator
	circuits in order to prove the first \emph{average-case lower bounds} 
	and design the first \emph{circuit analysis algorithms} for small comparator circuits.
    Developing circuit analysis algorithms is a crucial step for
    understanding a given circuit class~\cite{oliveira_algorithms_2013, williams_algorithms_2014}, 
    and are often obtained only after
    lower bounds have been proven for the class\footnote{One exception is
        $\ACC$ circuits, for which satisfiability algorithms are
        known~\cite{Wil14b}, and the only exponential lower bound known for
        $\ACC$ is a consequence of this algorithm. However, the function
        used in the lower bound is not in $\NP$.}.
    Many well-studied circuit classes have been investigated
    from 
    this perspective,
    such as~$\AC^0$
	circuits~\cite{impagliazzo_satisfiability_2012}, 
    De Morgan
    formulas~\cite{DBLP:journals/eccc/Tal15},
	branching programs~\cite{DBLP:journals/jacm/ImpagliazzoMZ19},
    $\ACC$ circuits~\cite{Wil14b},
    and many others
    (see also~\cite{DBLP:journals/cc/ChenKKSZ15,DBLP:conf/innovations/ServedioT17,
        DBLP:conf/coco/KabanetsKLMO20}).
    Our paper
    commences 
    an investigation of this kind 
    for comparator circuits.
	
	\subsection{Results}
	
	\noindent\textbf{Average-case lower bounds.}
    Our work starts with the first
    average-case lower bound against comparator
	circuits.
	
	\begin{theorem}[Average-case Lower Bound]\label{thm:lower_bound-intro}
		There exist constants $c,d\geq 1$ such that the following holds.
		For any $k\geq c\cdot \log n$,
		there is a polynomial-time computable function $f_k$ such that, for
		every comparator circuit $C$ with at most
		\[
		\frac{n^{1.5}}{d\cdot k\cdot \sqrt{\log n}}
		\]
		gates,	we have
		\[
		\Prob_{x\in\bool^n}[f_k(x)=C(x)]\leq \frac{1}{2} + \frac{1}{2^{\Omega(k)}}.
		\]
	\end{theorem}
	
    An important feature of the lower bound in \Cref{thm:lower_bound-intro}
    is that it matches the $\Omega\!\left((n/\log n)^{1.5}\right)$ worst-case
    lower bound of \cite{DBLP:conf/innovations/GalR20}, in the sense that we
    can recover their result (up to a multiplicative constant) by setting
    $k=O(\log n)$.

	Using ideas from the proof of the above average-case lower bound, we also
	show average-case lower bounds against various models that tightly
	match their state-of-the-art worst-case lower bounds, such as general
	formulas, (deterministic-, nondeterministic-, parity-)branching
	programs and span programs (see \Cref{sec:generalization}). 
	Note that strong average-case
	lower bounds against $n^{2-o(1)}$-size general formulas and
	deterministic branching programs were previously known
    \cite{DBLP:conf/stoc/KomargodskiR13,DBLP:journals/cc/ChenKKSZ15} but
    they did not match the worst-case lower bounds, 
    whereas tight average-case lower bounds for \emph{De Morgan} formulas
    were proved by~\cite{DBLP:journals/siamcomp/KomargodskiRT17}.
	
	\vspace{0.2cm}
	\noindent\textbf{$\#$SAT algorithms.} 
	The design of
	algorithms for interesting \emph{circuit analysis problems}
	is 
	a growing 
    line of research in circuit
    complexity~\cite{williams_algorithms_2014}. These are
	problems that take circuits as inputs. A famous example of such a
	circuit analysis problem is the \emph{satisfiability} problem (SAT),
	which asks to determine whether a given circuit has a satisfying
	assignment. Note that the satisfiability problem for polynomial-size
	general circuits is NP-complete, so it is not believed to have a
	polynomial-time (or subexponential-time) algorithm. However, one can
	still ask whether we can obtain \emph{non-trivial} SAT algorithms
	running faster than exhaustive search, say in time $2^n/n^{\omega(1)}$
	where $n$ is the number of variables of the input circuit, even for
	restricted circuit classes.
	While designing non-trivial SAT algorithms is an interesting problem by
	itself, it turns out that this task is also tightly connected to
	proving lower bounds. In particular, recent works of
	Williams~\cite{Wil13, Wil14b} 
	have shown
	that such a non-trivial
	satisfiability algorithm for a given class of circuits can often be
	used to prove non-trivial circuit lower bounds against that same circuit class.
	
	Here, we show an algorithm with non-trivial running time that counts
	the number of satisfying assignments of a given comparator circuit.
	
	\begin{theorem}[$\#$SAT Algorithms]\label{thm:main_sat-intro}
		There is a constant $d>1$ and a deterministic algorithm such that,
		for every $k\leq n/4$, given a comparator circuit on $n$ variables with at most
		\[
		\frac{n^{1.5}}{d\cdot k\cdot \sqrt{\log n}}
		\]
		gates, the algorithm outputs the number of satisfying assignments of $C$ in time
		\[
		2^{n - k}\cdot\poly(n).
		\]
	\end{theorem}
	
	Note that the running time in \Cref{thm:main_sat-intro} is non-trivial
	for 
	size up to $o\!\left(n/\log n\right)^{1.5}$, in which case $k=\omega(\log n)$
	and the running time becomes $2^n/n^{\omega(1)}$.
	
	\vspace{0.2cm}
	
	\noindent\textbf{Pseudorandom generators and ${\sf MCSP}$ lower bounds.}
	Another important circuit analysis problem is \emph{derandomization},
	which, roughly speaking, asks to decide whether a given circuit accepts
	or rejects a large fraction of its inputs. A standard approach to solve
	this problem is to construct a pseudorandom generator (PRG). A PRG
	against a class $\mathcal{C}$ of circuits is an efficient and
	deterministic procedure $G$ mapping short binary strings (seeds) to
	longer binary strings,
	with the property that
	$G$'s output (over uniformly random
	seeds) ``looks random'' to every circuit in $\mathcal{C}$. More
	precisely, we say that a generator $G\colon\bool^r\to\bool^n$
	\emph{$\varepsilon$-fools} a class $\mathcal{C}$ of circuits if, for
	every $C\colon\bool^n\to\bool$ from $\mathcal{C}$, 
	we have
	\[
	\left| \Prob_{z\in\bool^r}[C(G(z))=1] - \Prob_{x\in\bool^n}[C(x)=1]  \right|\leq \varepsilon,
	\]
	In constructing PRGs, we aim to minimize the parameter $r$, which is called the seed length.
	
	We show a PRG against comparator circuits of size $s$ with seed length $s^{2/3+o(1)}$.
	
	\begin{theorem}[Pseudorandom Generators]\label{prg-intro}
		For every $n \in \mathbb{N}$, $s=n^{\Omega(1)}$, and $\varepsilon \geq 1/\poly(n)$, there is a pseudorandom generator $G\colon \bool^{r}\to\bool^{n}$, with seed length
		\[
		r=s^{2/3 + o(1)},
		\]
		that $\varepsilon$-fools comparator circuits on $n$ variables with $s$ gates.
	\end{theorem}
	
	Note that the seed length of the PRG in \Cref{prg-intro} is non-trivial (i.e., $o(n)$) for comparator circuits of size $n^{1.5-o(1)}$.
	
	The PRG above has an application in obtaining lower bounds for the \emph{minimum circuit size
		problem} (${\sf MCSP}$) against comparator circuits. 
        The ${\sf MCSP}$ problem asks if a given truth table\footnote{A truth
		table is a bit-string that stores the output values of a Boolean function for
		all possible inputs.} represents a function that can be computed by some
	small-size circuit. Understanding the exact complexity of ${\sf MCSP}$ is a
	fundamental problem in complexity theory. Motivated by a recent line of
	research called \emph{hardness
        magnification}~\cite{DBLP:journals/eccc/OliveiraS18a,DBLP:conf/coco/OliveiraPS19,DBLP:conf/focs/ChenJW19,chen_beyond_2020},
        which states that a weak
	circuit lower bound for certain variants of ${\sf MCSP}$ implies
	breakthrough results in circuit complexity, researchers have been
	interested in showing lower bounds for ${\sf MCSP}$ against restricted
	classes of circuits. For many restricted circuit classes such as
	constant-depth circuits, formulas and branching programs, the lower bounds
	that have been proved for ${\sf MCSP}$ essentially match the best known
	lower bounds that we have for any explicit functions
    \cite{DBLP:conf/icalp/GolovnevIIKKT19,DBLP:journals/toct/CheraghchiKLM20,DBLP:conf/coco/KabanetsKLMO20}.
    Here we obtain ${\sf MCSP}$ lower bounds against comparator circuits
    that nearly match the worst-case lower bounds.
	
	\begin{theorem}[${\sf MCSP}$ Lower Bounds]\label{thm:MCSP-intro}
        Let ${\sf MCSP}[n^{\alpha}]$ denote the problem of deciding whether
        a given $n$-bit truth table represents a function that can be
        computed by some general circuit of size at most $n^{\alpha}$.
        For any $\varepsilon>0$ and any $0<\alpha \leq 1-\varepsilon$,
        the ${\sf MCSP}[n^{\alpha}]$ problem
        does not have comparator circuits with $n^{1+\alpha/2 - \varepsilon}$ gates.
	\end{theorem}

    Previously, non-trivial comparator circuit lower bounds
    were known only for functions satisfying Ne\v{c}iporuk's
    criterion~\cite{nechiporuk_1966,DBLP:conf/innovations/GalR20}, such as
    \emph{Element Distinctness} and \emph{Indirect Storage Access}.
    \Cref{thm:MCSP-intro}
    provides yet another natural computational problem which
    is hard for bounded-size comparator circuits.
    We remark that the $\MCSP$ problem is expected to require much larger
    circuits than the lower bound of \Cref{thm:MCSP-intro} provides;
    however, the lack of combinatorial, algebraic or
    analytic structure in the $\MCSP$ function means that proving lower
    bounds for it is usually hard.

	Finally, we also observe that the framework developed in
	\cite{DBLP:conf/innovations/ServedioT17} can be used to obtain a
	\emph{non-trivial} (distribution-independent) PAC learning algorithm
	for comparator circuits of size $n^{1.5-o(1)}$, that uses membership
	queries (see \Cref{sec:learning}).
	
	\subsection{Techniques}\label{sec:techniques}

	Random restrictions
	have been very fruitful
	in the study of weaker circuit classes, such 
    as~$\AC^0$
	circuits~\cite{haastad_almost_1986, impagliazzo_satisfiability_2012}, 
	De Morgan formulas~\cite{DBLP:conf/stoc/KomargodskiR13}
	and branching programs~\cite{DBLP:journals/jacm/ImpagliazzoMZ19},
	both for the proof of lower bounds
	and the construction of algorithms~\cite{DBLP:journals/cc/ChenKKSZ15}.
	However, as observed by Gál and
	Robere~\cite{DBLP:conf/innovations/GalR20}, 
	there are technical challenges when trying to apply this approach
	to comparator circuits.
	In this work, we successfully apply the method of random restrictions to
	comparator circuits for the first time.
	
	\vspace{0.2cm}
	
	\noindent\textbf{Average-case lower bounds.}
	At a high level, the proof of our average-case lower bound is based on
	the approach developed in
	\cite{DBLP:conf/stoc/KomargodskiR13,DBLP:journals/cc/ChenKKSZ15}, which
	can be used to obtain average-case-lower bounds against circuits that
	admit a property called ``shrinkage with high probability under random
	restrictions''. Roughly speaking, this property says that, if we
	randomly fix the values of some variables in the circuit except for a
	$0<p<1$ fraction of them, then its size shrinks by a factor of
	$p^{\Gamma}$ for some $\Gamma>0$, with very high probability. This
	method has been used to obtain strong average-case lower bounds against
	$n^{2.5-o(1)}$-size De Morgan formulas
	\cite{DBLP:conf/stoc/KomargodskiR13,DBLP:journals/cc/ChenKKSZ15} (later
	improved to $n^{3-o(1)}$ by
	\cite{DBLP:journals/siamcomp/KomargodskiRT17}) and $n^{2-o(1)}$-size
	general formulas and deterministic branching programs.
	
	An obvious issue of applying this approach to comparator circuits is
	that we don't know how to shrink the size (i.e., number of gates) of a
	comparator circuit using random restrictions, as when we fix the value
	of a (non-trivial\footnote{We say that a wire is non-trivial if it is
		connected to some gate.}) wire, we may only be able to remove one gate
	in the worst scenario (i.e., the gate that is directly connected to that
	wire). The idea is that instead of shrinking the \emph{number of
		gates}, we will try to shrink the \emph{number of wires}. The reason
	why this can be useful is that \emph{one can effectively bound the
		number of gates of a comparator circuit by its number of wires}; this
	is a structural result of comparator circuits proved by G{\'{a}}l and
	Robere \cite{DBLP:conf/innovations/GalR20} and was the key ingredient
	in proving their worst-case lower bound. More precisely, they showed
	that any comparator circuit that has at most $\ell$ wires needs no more
	than $\ell^2$ gates (see \Cref{lem:structural}). Now following
	\cite{DBLP:conf/stoc/KomargodskiR13,DBLP:journals/cc/ChenKKSZ15}, one
	can show that under some certain type of random restriction that leaves
	a $p\vcentcolon=k/n$ fraction of the variables unfixed, for any large
	enough $k$, the number of wires of a comparator circuit will shrink
	(with very high probability) by roughly a factor of $p$, and hence its
	number of gates is bounded by $(p\cdot \ell)^2$. By letting
	$\ell=o\!\left(n^{1.5}/\!\left(k\cdot \sqrt{\log n}\right)\right)$, this
	size is less than $o\!\left(n/\log n\right)$ and from there one can
	show that the original circuit cannot compute some hard function on more
	than $1/2+1/2^{k^{\Omega(1)}}$ fraction of the inputs.
	
	While the above gives an average-case lower bound, it does not match
	the worst-case one, because we need to set $k\geq \log^c n$ for some
	(unspecified) constant $c>1$, which is controlled by the type of random
	restrictions and the extractor used in the construction of the hard
	function in both
	\cite{DBLP:conf/stoc/KomargodskiR13,DBLP:journals/cc/ChenKKSZ15}. This
	means we can only achieve a lower bound that is at best $n^{1.5}/(\log
	n)^{c+.5}$ (even for worst-case hardness). In order to be able to set
	$k=O(\log n)$, one way is to use a more sophisticated (so called non-explicit
	bit-fixing) extractor shown in
	\cite{DBLP:journals/siamcomp/KomargodskiRT17}, which will allow us to
	set $k\in \left[O\!\left(\log n\right),
	\Omega\!\left(n^{1/3}\right)\right]$ (with hardness
	$1/2+1/2^{\Omega(k)}$). Here we refine and simplify 
	this
	approach
	in the case of comparator circuits by using a more structural
	(block-wise) random restriction that shrinks the number of wires with
	probability one.
	Such a random restriction, when combined with a
	simple extractor, allows us to set $k\in \left[O\!\left(\log n\right),
	\Omega\!\left(n\right)\right]$.
	
	\vspace{0.2cm}
	\noindent\textbf{$\#$SAT algorithms.}
	Based on the above analysis in showing average-case lower bounds, one can try to design a SAT algorithm for comparator circuits in a way that is similar to that of \cite{DBLP:journals/cc/ChenKKSZ15}, which combines ``shrinkage under restrictions'' with a memorization technique. Suppose we have a comparator circuit $C$ with $s\vcentcolon=o\!\left(n^{1.5}/(k\cdot \sqrt{\log n})\right)$ gates and and $\ell\leq s$ non-trivial wires. By partitioning the variables into $n/k$ equal-size blocks, we can show that there is some block $S_i$ such that after fixing the variables outside of this block, the number of wires in the restricted circuit is at most $\ell_0\vcentcolon=\ell/(n/k)\leq o\!\left(\sqrt{n/\log n}\right)$. Again by the structural property of comparator circuits (\Cref{lem:structural}), this restricted circuit, which is on $k$ variables, has an equivalent circuit with $o\!\left(n/\log n\right)$ gates. Then to count the number of satisfying assignments for the original circuit, we can first memorize the numbers of satisfying assignments for all circuits with at most with $o\!\left(n/\log n\right)$ gates. There are $2^{o(n)}$ of them and hence we can compute in time $2^k\cdot 2^{o(n)}$ a table that stores those numbers. We then enumerate all possible $2^{n-k}$ restrictions $\rho\in\bool^{[n]\backslash S_i}$ and for each $\rho$ we look up the number of satisfying assignments of the restricted circuit $\rst{C}{\rho}$ from the pre-computed table. Summing these numbers over all the $\rho$'s gives the number of satisfying assignments of $C$. 
	
	However, there is a subtle issue in the above argument: although we
	know that a restricted circuit has an equivalent simple circuit with
	$o\!\left(n/\log n\right)$ gates, we do not know which simple circuit it
	is equal to. Note that when we fix the value of a (non-trivial) wire,
	we may only be able to remove one gate, so the number of gates left in
	the restricted circuit is possibly $s-(\ell-\ell_0)$, which can be much
	larger than $n/\log n$, and it is not clear how we can further simplify
	such a circuit \emph{efficiently}. To overcome this issue, we explore
	structural properties of comparator circuits to show how to construct
	a more sophisticated data structure that not only can tell us the
	number of satisfying assignments of a circuit with $o\!\left(n/\log
	n\right)$ gates but also allows us to \emph{efficiently} simplify each restricted circuit to an equivalent circuit with at most this many gates.
	
	\vspace{0.2cm}
	\noindent\textbf{Pseudorandom generators and ${\sf MCSP}$ lower bounds.}
	Our PRG against comparator circuits builds upon the paradigm of
	\cite{DBLP:journals/jacm/ImpagliazzoMZ19}, which was used to construct
	PRGs against circuits that admit ``shrinkage under
	\emph{pseudorandom} restrictions''. As in the proof of our
	average-case-lower bound, in order to apply this paradigm, we will 
	shrink the number of wires
	instead of the number of gates. Following
	\cite{DBLP:journals/jacm/ImpagliazzoMZ19}, we prove a pseudorandom
	shrinkage lemma for comparator circuits, which can then be used to
	obtain a PRG of seed length $s^{2/3+o(1)}$, where $s$ is the size of a
	comparator circuit.
	
	As observed in \cite{DBLP:journals/toct/CheraghchiKLM20}, one can
	modify the construction of the PRG in
	\cite{DBLP:journals/jacm/ImpagliazzoMZ19} to make it ``locally
	explicit''. This means that, for every seed, the output of the PRG, when
	viewed as a truth table of a function, has circuit complexity that is
	about the same as the seed length. Such a ``local'' PRG immediately
	implies that ${\sf MCSP}$ cannot be
    computed by comparator circuits of size $n^{1.5-o(1)}$, when the size
    parameter of ${\sf MCSP}$ is nearly-maximum (i.e., $n/O(\log n)$)
    \footnote{Note that ${\sf MCSP}$ takes two input parameters: a truth
    table and a size parameter $\theta$, and asks whether the given truth
table has circuit complexity at most $\theta$.}. Furthermore, we
	show a better trade-off between the size parameters of ${\sf MCSP}$ and
	the lower bound size of the comparator circuits, as in \Cref{thm:MCSP-intro}. This
	is similar to what was done by \cite{DBLP:conf/stoc/ChenJW20} in the
	case of ${\sf MCSP}$ lower bounds against De Morgan formulas.

	\subsection{Directions and open problems}


    We now state some further directions and open problems for which
    our work may be a starting point, or that are connected to our results.
	
	\vspace{0.2cm}
    \noindent\textbf{Algorithms and lower bounds for larger comparator circuits.}
    Our lower bounds and circuit analysis algorithms
    only work for comparator circuits of size up to
    $n^{1.5-o(1)}$. Can we improve this?
    Specifically,
    can we show 
    a lower bound of $n^{1.51}$ for comparator circuits
    computing a function of $n$ bits,
    and design algorithms for
    comparator circuits of the same size?
    In this paper, we used the
    random restriction method to analyse comparator circuits by shrinking
    the number of wires and using a structural result 
    of~\cite{DBLP:conf/innovations/GalR20} that relates the number of gates to
    the number of wires. Can we analyse the effect of random restrictions
    on the gates \emph{directly}, and show a shrinkage lemma for comparator
    circuits on the number of gates, with a shrinkage exponent
    $\Gamma>1/2$? Such a lemma would imply a lower bound that is better than
    $n^{1.5}$, and would allow us to design algorithms for
    comparator circuits larger than $n^{1.5}$.
		
	\vspace{0.2cm}
	\noindent\textbf{Hardness magnification near the state-of-the-art.}
	Recent work on \emph{hardness
		magnification}~\cite{DBLP:journals/eccc/OliveiraS18a,DBLP:conf/coco/OliveiraPS19,DBLP:conf/focs/ChenJW19,chen_beyond_2020}
	has shown that marginally improving the state-of-art worst-case lower bounds
	in a variety of circuit models
	would imply major breakthroughs in complexity theory.
	Although we don't prove this here, it is possible to show
	hardness magnification results for comparator circuits of size
	$n^{2+o(1)}$ by a simple adaptation of their arguments.
	Unfortunately, this does not match the best lower bounds we have for
	comparator circuits, which are around $n^{1.5-o(1)}$ as we have seen.
	Can 
	we
	show a hardness magnification phenomenom
	nearly matching the state-of-art lower bounds for comparator circuits?

	\vspace{0.2cm}
	\noindent\textbf{Extensions and restrictions of comparator circuits.}
	Recent work of Komarath, Sarma and Sunil~\cite{DBLP:journals/iandc/KomarathSS18}
	has provided characterisations of various complexity classes, such as
	$\L, \P$ and $\NP$,
	by means of extensions or restrictions of comparator circuits.
	Can our results and techniques applied to comparator circuits be
	extended to those variations of comparator circuits? Can this extension
	shed any light into the classes characterised
	by~\cite{DBLP:journals/iandc/KomarathSS18}?
	%

\bnote{I think it is better not to include the branching programs problem.. It has no connection to
anything we contribute. It has more place in a survey, or maybe inside the
introductory introduction (as we did on the discussion about $\maj$).}

	\section{Preliminaries}
	\subsection{Definitions and notations}

	For $n\in\mathbb{N}$, we denote $\{1,\dots,n\}$ by $[n]$.
	For a string $x$, we denote by $\mathrm{K}(x)$ the \emph{Kolmogorov
		complexity} of $x$, which is defined as the minimum length of a Turing
	machine that prints $x$ as output.
	
	\vspace{0.2cm}
	
	\noindent\textbf{Restrictions. } A \emph{restriction} for an
    $n$-variate Boolean function $f$, denoted by $\rho \in \set{0,1,*}^n$,
    specifies a way of fixing the values of some subset of variables for
    $f$. That is, if $\rho(i)$ is $*$, we leave the $i$-th variable
    unrestricted and otherwise fix its value to be $\rho(i) \in\bool$. We
    denote by $\rst{f}{\rho} : \blt^{\rho^{-1}(*)} \to \blt$ the restricted
    function after the variables are restricted according to $\rho$, where
    $\rho^{-1}(*)$ is the set of unrestricted variables.
	
	\vspace{0.2cm}
	
	\noindent\textbf{Comparator circuits. }	We define comparator circuits
    as a set of \emph{wires}
	labelled by an input literal (a variable $x_i$ or its negation
    $\neg x_i$), a sequence of \emph{gates},
	which are \emph{ordered} pairs of wires,
    and a designated \emph{output wire}.
	In other words, each gate is a pair of wires $(w_i,w_j)$,
	denoting that the wire $w_i$ receives the logical conjunction
	($\land$) of the wires, and
	$w_j$ receives the logical disjunction
	$(\lor)$.
    On a given input $a$,
    a comparator circuit computes as follows:
    each wire labelled with a literal $x_i$
    is initialised with $a_i$,
    and we update the value of the wires
    by following the sequence of gates;
    the output wire contains the result of the computation.
	A wire is called \emph{non-trivial} if there is a gate connected to
	this wire.
	Note that, if a comparator circuit has $\ell$ non-trivial wires and
	$s$ gates, then $\ell \leq s$.
	This means that lower bounds on the number of wires also imply lower bounds on
	the number of gates.
	
	\subsection{Structural properties of comparator circuits}

	For a gate $g$ in a comparator circuit and an input $x\in\bool^n$, we denote by $u_g(x)$ (resp. $v_g(x)$) the first (resp. second) in-value to the gate $g$ when given $x$ as input to the circuit.
	
	\begin{definition}[Useless Gates]\label{useless}
		We say that a gate $g$ in a comparator circuit is \emph{useless} if either one of the following is true:
		\begin{enumerate}
			\item for every input $x$, $(u_g(x),v_g(x))\in\{
			(0,1), (0,0), (1,1)\}$.
			\item for every input $x$, $(u_g(x),v_g(x))\in\{
			(1,0), (0,0), (1,1)\}$.
		\end{enumerate}
		We say that a useless gate is of \emph{TYPE-1} (resp. \emph{TYPE-2}) if it is the first (resp. second) case. Also, a gate is called \emph{useful} if it is not useless. 
	\end{definition}
	
	The following proposition allows us to remove useless gates from a comparator circuit.
	\begin{proposition}[{\cite[Proof of Proposition 3.2]{DBLP:conf/innovations/GalR20}}]
		\label{removing}
		Let $C$ be a comparator circuit whose gates are $g_1,g_2,\dots,g_s$ (where $g_s$ is the output gate) and let $g_i=(\alpha,\beta)$ be any useless gate in $C$.
		\begin{itemize}
			\item Suppose $g_i$ is of TYPE-1. Then the circuit $C'$ obtained from $C$ by removing the gate $g_i$ computes the same function as that of $C$.
			\item Suppose $g_i$ is of TYPE-2. Let $C'$ be the circuit whose gates are $g_1,g_2,\dots,g_{i-1},g'_{i+1}\dots,g'_s$, where for $j=i+1,\dots,s$, $g'_j$ is obtained from $g_j$ by replacing $\alpha$ with $\beta$ (if $g_j$ contains $\alpha$) and at the same time replacing $\beta$ with $\alpha$ (if $g_j$ contains $\beta$). Then $C'$ computes the same function as that of $C$.
		\end{itemize}
	\end{proposition}
	\begin{proof}
		On the one hand, if $g$ is a TYPE-1 useless gate, then for every input to the circuit, the out-values of $g$ are the same as its in-values, so removing $g$ does not affect the function computed by the original circuit. On the other hand, if $g$ is of TYPE-2, then the in-values feeding to $g$ will get swapped after $g$ is applied. This has the same effect as removing $g$ and ``re-wiring'' the gates after $g$ so that a gate connecting one of the wires of $g$ gets switched to connect the other wire of $g$, as described in the second item of the proposition.
	\end{proof}
	
	We need the following powerful structural result for comparator circuits from \cite{DBLP:conf/innovations/GalR20}.  	
	\begin{theorem}[{\cite[Theorem 1.2]{DBLP:conf/innovations/GalR20}}]\label{thm:GR}
		If $C$ be is a
        comparator circuit with $\ell$ wires and $s$ gates such that every gate in $C$ is useful, then $s\leq \ell\cdot(\ell-1)/2$.
	\end{theorem} 
	
	\Cref{removing} and \Cref{thm:GR} together give the following lemma. 
	\begin{lemma}\label{lem:structural}
        Every comparator circuit with $\ell > 0$ wires has an equivalent comparator circuit with $\ell$ wires and with at most $\ell\cdot (\ell-1)/2$ gates.
	\end{lemma}

	\section{Average-case Lower Bounds}\label{sec:lower-bounds}
	In this section, we prove our average-case lower bound against comparator circuits. We first describe the hard function.
	
	\subsection{The hard function}
	
	\noindent\textbf{List-decodable codes. }
	Recall that a $(\zeta,L)$-list-decodable binary code is a function $\mathrm{Enc}\colon \bool^n \to \bool^m$ that maps $n$-bit messages to $m$-bit codewords so that, for each $y\in\bool^m$, there are at most $L$ codewords in the range of $\mathrm{Enc}$ that have relative hamming distance at most $\zeta$ from $y$. We will use the following list-decodable code.
	
	\begin{theorem}[See e.g., {\cite[Proof of Theorem 6.4]{DBLP:journals/cc/ChenKKSZ15}}] \label{ecc}
        There is a constant $c>0$ such that for any given $k=k(n)>c\cdot
        \log n$, there exists a binary code $\mathrm{Enc}$ mapping $n$-bit
        message to a codeword of length $2^{k}$, such that $\mathrm{Enc}$
        is $(\zeta, L)$-list-decodable for
        $\zeta=1/2-O\!\left(n/2^{k/2}\right)$ and $L\leq
        O\!\left(2^{k/2}/n\right)$. Furthermore, there is a polynomial-time
        algorithm for computing the $i$-th bit of $\mathrm{Enc}(x)$, for
        any inputs $x \in \bool^{n}$ and $i \in \left[2^{k}\right]$.
	\end{theorem}
	
	\begin{definition}[Generalized Andreev's Function]\label{def-And}
		Let $k$ be a positive integer. Define $A_{k}\colon\bool^{n+n}\to\bool$ as follow:
		\[
		A_{k}(x_1,\dots,x_{n},y_1,\dots,y_{n}) \vcentcolon= \mathrm{Enc}(x_1,\dots,x_{n})_{\alpha(y_1,\dots,y_{n})},
		\]
		where $\mathrm{Enc}$ is the code from \Cref{ecc} that maps $n$ bits to $2^{k}$ bits, and $\alpha\colon\bool^n\to\bool^k$ is defined as
		\[
		\alpha(y_1,\dots,y_{n}) \vcentcolon= \left(\bigoplus_{i=1}^{n/k} y_i, \bigoplus_{i=n/k+1}^{2n/k} y_i, \dots, \bigoplus_{i=(k-1)n/k + 1}^{n} y_i\right).
		\]
		That is, the function $\alpha$ partitions $y$ evenly into $k$ consecutive blocks and outputs the parities of the variables in each block.
	\end{definition}
	
	Note that the function $A_{k}$ defined above is polynomial-time computable since we can compute  $\alpha(y)$ and $\mathrm{Enc}(x)_i$ for any given $i$ in $\poly(n)$ time.

	\subsection{Proof of the average-case lower bound}
	We will show a lower bound on the number of wires, which automatically
    implies a lower bound on the number of gates.

	\begin{theorem}\label{thm:lower_bound}
		There exist constants $c,d\geq 1$ such that the following holds.
		For any $k\geq c\cdot \log n$,
		there is a polynomial-time computable function $f_k$ such that, for every comparator circuit $C$ whose number of wires is
		\[
		\frac{n^{1.5}}{d\cdot k\cdot \sqrt{\log n}},
		\]
		we have
		\[
		\Prob_{x\in\bool^n}[f_k(x)=C(x)]\leq \frac{1}{2} + \frac{1}{2^{\Omega(k)}}.
		\]
	\end{theorem}
	\begin{proof}
		Let $A_{k}$ be the generalized Andreev's function on $2n$ variables. Let $C$ be a comparator circuit on $2n$ variables with $\ell \leq n^{1.5}/ \!\left(d\cdot k\cdot \sqrt{\log n}\right)$ wires, where $d\geq 1$ is a sufficiently large constant. To avoid some technicalities due to divisibility that can be overcome easily, we assume that $n$ is divisible by $k$.

		We need to upper bound the following probability.
		\begin{align*}
		\Prob_{x,y\in\bool^{n}\times \bool^{n}}[A_{k}(x,y)=C(x,y)] &\leq \Prob_{x,y}[A_{k}(x,y)=C(x,y) \mid \mathrm{K}(x)\geq n/2]  +\Prob_{x}[\mathrm{K}(x)< n/2]\\
		&\leq \Prob_{x,y}[A_{k}(x,y)=C(x,y) \mid \mathrm{K}(x)\geq n/2]+ \frac{1}{2^{n/2}}.
		\end{align*}
		Let $x$ be any fixed $n$-bit string with Kolmogorov complexity at least $n/2$. Let $A'\colon\bool^n\to\bool$ be
		\[
		A'(y) \vcentcolon=A_{k}(x,y), 
		\]
		and let $C'$ be a comparator circuit on $n$ variables with at most $\ell$ wires defined as
		\[
		C'(y) \vcentcolon=C(x,y).
		\]
		We will show that
		\[
		\Prob_{y\in\bool^{n}}\left[A'(y)=C'(y)\right]\leq \frac{1}{2}+ \frac{n}{2^{k/4}}.
		\]
		
		First of all, let us divide the $n$ variables of $C'$ into $n/k$
		parts, each of which contains $k$ variables, as follows. We first
		partition the $n$ variables evenly into $k$ \emph{consecutive}
		blocks, denoted as $B_1,B_2,\dots,B_k$. Then we define the $i$-th
		part $S_i$, where $i\in[n/k]$, to be the union of the $i$-th
		variables in each of $B_1,B_2,\dots,B_k$. That is
		\[
		S_i \vcentcolon= \bigcup_{j\in[k]} \left\{y\colon y \text{ is the $i$-th variables of } B_j\right\}.
		\]
		Now we count the number of wires that are labelled by the variables in each $S_i$ and let
		\[
		w_i \vcentcolon= \left|\{u \colon \text{$u$ is a wire labelled by
        some $x\in S_i$ (or its negation)}\}\right|.
		\]
		We have
		\[
		\sum_{i\in[n/k]} w_i = \ell,
		\]
		which implies that there is a particular $i\in[n/k]$ such that
		\[
		w_{i} \leq \frac{\ell}{n/k} \leq\frac{1}{d} \cdot \sqrt{\frac{n}{\log n}} =\vcentcolon \ell_0.
		\]
		
		Next, we will consider restrictions that fix the values of the
		variables outside $S_i$. Note that if we fix the value of a variable
		$x_i$ in a comparator circuit, then we can obtain a restricted circuit
		so that all the wires that are labelled by either $x_i$ or $\neg{x_i}$
		are eliminated, after some appropriate updates on the gates in
		the circuit. This is not an obvious fact. One way to see this is that
		once we fix the value of a wire, the gate that directly connects this
		wire becomes \emph{useless} in the sense of \Cref{useless} so it can be
		removed after some appropriate ``re-wirings'' of the gates in the circuit as described in \Cref{removing}. Then we can keep doing this until no gate is connected to that wire, in which case the wire can be removed from the circuit.
		
		Now we have
		\[
		\Prob_{y\in\bool^{n}}\left[A'(y)=C'(y)\right] = \Prob_{\rho\in\bool^{[n]\backslash S_i},z\in\bool^{k}}\left[\rst{A'}{\rho}(z)=\rst{C'}{\rho}(z)\right].
		\]
		It suffices to upper bound
		\[
		\Prob_{z\in\bool^{k}}\left[\rst{A'}{\rho}(z)=\rst{C'}{\rho}(z)\right],
		\]
		for every $\rho\in\bool^{[n]\backslash S_i}$.
		For the sake of contradiction, suppose for some $\rho$, we have
		\begin{equation}\label{ext-property}
		\frac{1}{2} +\frac{n}{2^{k/4}} < \Prob_{z\in\bool^{k}}\left[\rst{A'}{\rho}(z)=\rst{C'}{\rho}(z)\right]
		= \Prob_{z\in\bool^{k}}\left[\mathrm{Enc}(x)_{\alpha}=\rst{C'}{\rho}(z)\right],
		\end{equation}
		where $\alpha\in\bool^k$ is 
		\[
		\alpha_j \vcentcolon= 
		\parity\!\left(\rho|_{B_j \backslash S_i}\right) \oplus z_j,
		\]
		and $\rho|_{B_j \backslash S_i}$ denotes the partial assignment
		given by $\rho$ but restricted to only variables in the set $B_j
		\backslash S_i$.
		Note that $\alpha$ is uniformly distributed for uniformly random $z$.
		Therefore, if we have the values of	$\parity\!\left(\rho|_{B_j \backslash S_i}\right)$ 
		for each
		$j\in[k]$ ($k$ bits in total), and if we know the restricted
		circuit $\rst{C'}{\rho}$, then we can compute the codeword
		$\mathrm{Enc}(x)$ correctly on at least $1/2 +n/ 2^{k/4}$
		positions, by evaluating $\rst{C'}{\rho}(z)$ for every $z\in\bool^k$. As
		a result, we can list-decode $\mathrm{Enc}(x)$, and, using additional
		$k/2$ bits (to specify the index of $x$ in the list), we can
		recover $x$ exactly. Finally, note that the number of wires in
		$\rst{C'}{\rho}$ is at most $\ell_0$. Therefore, by Lemma~\ref{lem:structural}, such a
		circuit can be described using a string of length at most
		\begin{align*}
		O\!\left(\ell_0\cdot\log(n) + \ell_0^2\cdot \log(\ell_0)\right) &\leq O\!\left(\ell_0^2\cdot \log n\right)\\
		&= O\!\left(\frac{n}{d^2\cdot \log n}\cdot \log n\right)\\
		&\leq n/4,
		\end{align*}
		where the last inequality holds when $d$ is sufficiently large.
		Therefore, we can recover $x$ using less than 
		\[
		n/4 + k + k/2 + O(\log n) < n/2
		\]
		bits. 
		Here we assume $k \leq n/8$ since otherwise the theorem can be shown trivially. 
		This contradicts the fact that the Kolmogorov complexity of $x$ is at least $n/2$.
	\end{proof}
	
	\section{Tight Average-case Lower Bounds from a Ne\v{c}iporuk-Type Property}\label{sec:generalization}
	
	Here, we describe a generalization of the average-case lower bound in \Cref{sec:lower-bounds} to circuit classes whose worst-case lower bounds can be proved via Ne\v{c}iporuk's method.
	
    \begin{theorem}\label{thm:general}
        There is a constant $c>1$ such that the following holds. Let
        $\mathcal{C}$ be a class of Boolean circuits that is closed under
        restrictions.  Suppose that, for any 
        $k \in [c \cdot \log n,\,n/3]$,
        there
        exists a partition of the $n$ variables into $m\vcentcolon=n/k$
        equal-sized blocks $S_1,S_2,\dots,S_{m}$ and a collection of
        $k$-input-bit functions $\mathcal{H}$ such that
        \begin{enumerate}
            \item 
                $|\mathcal{H}|\leq 2^{n/2}$, and
            \item 
                for every $C\in \mathcal{C}_n$ of size
                $s(n,k)$, there exists some block $S_i$ such that 
                $
                \left\{\rst{C}{\rho}\right\}_{\rho\in\bool^{[n]\backslash
                S_i}}\subseteq\mathcal{H}.
                $
        \end{enumerate}
        \bnote{What is $s(n,k)$?}
        \znote{$s(n,k)$ is the size of the circuits, which can depend on
        the parameters $n$ and $k$. $k$ can be thought as a parameter that
        controls the quality of the average-case lower bounds. In the case of
        comparator circuits, we have $s(n,k)=n^{1.5}/(k\cdot \sqrt{\log n})$.}
        Then for any 
        $k \in [c \cdot \log n,\,n/6]$,
        there exists a
        polynomial-time computable function $f_{k}$ which satisfies
        \[
            \Prob_{x\in\bool^n}[C(x)=f_{k}(x)]\leq \frac{1}{2}+\frac{1}{2^{\Omega(k)}},
        \]
        for every $C\in \mathcal{C}_n$ of size $s(n/2,k)$.
        \bnote{Why the change from $s(n,k)$ to $s(n/2,k)$?}
        \\
        \znote{This is some technicality in the proof. We show lower bounds on circuits over $2n$ variables, but the assumption is on circuits over $n$ bits.}
        
        \znote{The theorem has two parts. One starts with ``Suppose....'' and the other starts with ``Then for...''. The parameters $n$ in different parts do not mean the same quantity. $n$ in the second part is twice the $n$ in the first part. (that's why in the second part $s(n/2,k)$ is equal to $s(n,k)$ in the first part and $k$ is less than $n/3$ in the first part but becomes $n/6$ in the second part) Maybe we should use different letters.}
    \end{theorem}

	\noindent \textbf{Remark.}
    In the original Ne\v{c}iporuk's argument for getting worst-case lower
    bounds, it is only required that, for every $C\in \mathcal{C}$, there is
    some block such that the number of distinct functions, after fixing the
    variables outside of the block, is at most $2^{n/2}$, and \emph{this set
    of functions can be different for different $C$}. For
    \Cref{thm:general}, we need something stronger which says that it is
    the \emph{same set of $2^{n/2}$ functions for every $C$}.
    We remark that, though the weaker condition is sufficient for
    \emph{worst-case} lower bounds,
    all applications of Ne\v{c}iporuk's method known to us
    also prove the stronger condition,
    thus yielding \emph{average-case} lower bounds by~\Cref{thm:general}.
	
	\vspace{0.2cm}
	
	\Cref{thm:general} requires a slightly different argument than that of \Cref{thm:lower_bound}. Its proof is presented in \Cref{sec:generalization-proof}.
	
	By combining \Cref{thm:general} with known structural properties for various models (see e.g., \cite{DBLP:books/daglib/0028687}), we get that for the class of circuits $\mathcal{C}$ of size $s$, where
	\begin{itemize}
		\item $\mathcal{C}$ is the class of general formulas, and $s=n^2/O(k)$, or
		\item $\mathcal{C}$ is the class of deterministic branching programs or switching networks,
		and $s=n^2/O(k\cdot\log n)$, or
		\item $\mathcal{C}$ is the class of nondeterministic branching programs, parity branching programs, or span programs, and $s=n^{1.5}/O(k)$,
	\end{itemize}
    there exists a function $f_k$ such that
    $\Prob_{x\in\bool^n}[C(x)=f_{k}(x)]\leq 1/2+1/2^{\Omega(k)}$ for every
    $C\in \mathcal{C}$,
	which matches the state-of-the-art worst-case lower bounds (up to a multiplicative constant) by letting $k=O(\log n)$.

	\section{\texorpdfstring{$\#$SAT}{\#SAT} Algorithms}
	In this section, we present our $\#$SAT algorithm for comparator
	circuits. As mentioned briefly in \Cref{sec:techniques}, we will need a
    preprocessed data structure that enables us to efficiently convert a
    circuit with small number of wires but large number of gates to an
    equivalent circuit (with the same number of wires) whose number of
    gates is at most quadratic in the number wires.
	
	\subsection{Memorization and simplification of comparator circuits}
	\begin{lemma}\label{DS}
		Let $n,\ell\geq 1$ be integers. For any fixed labelling of $\ell$ wires on $n$ variables, there is a data structure $\mathsf{DS}$ such that
		\begin{itemize}
			\item $\mathsf{DS}$ can be constructed in time $2^{n}\cdot \ell^{O(\ell^2)}$.
			\item Given access to $\mathsf{DS}$ and given any comparator circuit $C$ with $\ell$ wires (whose labelling is consistent with the one used for $\mathsf{DS}$) and $s$ gates, we can output in time $\poly(s,\ell)$ the number of satisfying assignments of $C$. Moreover, we obtain a comparator circuit with $\ell$ wires and at most $\ell\cdot(\ell-1)/2$ gates that is equivalent to $C$.
		\end{itemize}
	\end{lemma}
	\begin{proof}
		We know that every comparator circuit with $\ell$ wires has an equivalent circuit with $\ell\cdot(\ell-1)/2$ gates (\Cref{lem:structural}). Therefore, we can try to memorize the number of satisfying assignments for each of these circuits (by brute-force). Then for a given circuit $C$ with $\ell$ wires and $s$ gates where $s\gg \poly(\ell)$, we need to simplify $C$ to be a circuit with $\ell\cdot(\ell-1)/2$ gates so that we can look up its number of satisfying assignments, which was already computed. However, it is not clear how we can \emph{efficiently} simplify such a comparator circuit.
		
		The idea here is to remove the useless gates one by one (from left to right). To do this, firstly, we need to be able to tell whether a gate is useless, and secondly whenever we remove a useless gate, we need to ``re-wire'' the gates that come after that gate, which can depend on the types of the useless gate that we are removing, as described in \Cref{removing}.
		
        More specifically, $\mathsf{DS}$ will be a ``tree-like'' structure
        of depth at most $\ell\cdot(\ell-1)/2+2$ where the internal nodes are
        labelled as gates.
        Note that a path from
        the root to any internal node in the tree gives a sequence of gates, which
        specifies a comparator circuit up to a choice of the output wire.
        We will require the label of 
        every internal node to be a useful gate
        in the circuit specified by the path from the root to the node.
        In other words,
        each internal node will branch on all possible useful gates that
        could occur next in the circuit.
        Moreover, each leaf is either labelled as a useless gate, with respect to the circuit
        specified by the path from the root to the current leaf, or is
        labelled as a
        single wire that is designed to be the output wire. 

        For every leaf
        that is a useless gate, we store its type, and for each leaf that
        is a single wire, we store the number of satisfying assignments of
        the circuit that is specified by the path from the root to the
        leaf. Moreover, each internal node is a useful gate whose children are
        indexed by the set of all possible gates 
        (each is an ordered pair of wires) 
        and the set of wires
        (called an \emph{output leaf}).
        Note that checking whether a new gate is useless and
        computing its type require evaluating the current circuit on all
        possible inputs, which takes time $2^{n}\cdot\poly(\ell)$, but this
        is fine with our running time. Similarly, we can compute the number
        of satisfying assignments in each output leaf by brute force. Note
        that by \Cref{thm:GR}, the depth of such a tree is at most
        $\ell\cdot(\ell-1)/2+2$, otherwise there would be a comparator
        circuit with $\ell$ wires that has more than $\ell\cdot(\ell-1)/2$
        useful gates. Since each internal node has at most $\ell^2$ children, the
        tree has at most $\ell^{O\!\left(\ell^{^2}\right)}$ nodes in total.
        Since each node can be constructed in time
        $2^n\cdot\poly(\ell)$, the running time is clear.
		
		To look up the number of the satisfying assignments of a given circuit
		$C$ (with a labelling of the wires that is consistent with the one
		used for $\mathsf{DS}$), we start from the root of $\mathsf{DS}$, and
		move down the tree as we look at the gates in $C$ one by one (from
		left to right in the natural way). If we reach an output leaf, we
		output the number of satisfying assignments stored in that leaf. 
        However,
        if we reach a leaf $v$ that is specified as a useless
		gate, we remove the corresponding gate in $C$ and update the gates
		that come after it according to the type of this useless gate, using
		\Cref{removing}. Once we update the circuit, we start again from the
		parent of $v$ and look at the next gate in the updated circuit. We
		repeat this until we reach an output leaf. 
	\end{proof}

	\subsection{The algorithm}
	We will show an algorithm for comparator circuits with small number of wires, while the number of gates can be polynomial.
	
	\begin{theorem}\label{thm:main_sat}
		There is a constant $d>1$ and a deterministic algorithm such that for every $k\leq n/4$, given a comparator circuit on $n$ variables with at most
		\[
		\frac{n^{1.5}}{d\cdot k\cdot \sqrt{\log n}}
		\]
		wires and $\poly(n)$ gates, the algorithm outputs the number of satisfying assignments of $C$ in time
		\[
		2^{n - k}\cdot\poly(n).
		\]
	\end{theorem}
	\begin{proof}
		Let $C$ be a comparator circuit with $\ell \leq n^{1.5}/ \!\left(d\cdot k\cdot \sqrt{\log n}\right)$ wires and $\poly(n)$ gates, where $d\geq 1$ is a sufficiently large constant.
		
        We partition the $n$
        variables almost-evenly into $\lfloor n/k \rfloor$ consecutive
        blocks, denoted as $S_1,S_2,\dots,S_{\lfloor n/k \rfloor}$. 
		We then count the number of wires that are labelled by the variables in each $S_i$ and let
		\[
		w_i \vcentcolon= \left|\{u \colon \text{$u$ is a wire labelled by
        some $x\in S_i$ (or its negation)}\}\right|.
		\]
		We have
		\[
		\sum_{i\in[\lfloor n/k \rfloor]} w_i = \ell,
		\]
		which implies that there is a particular $i\in[k]$ such that
		\[
		w_{i} \leq \frac{\ell}{\lfloor n/k \rfloor} \leq\frac{1}{d} \cdot \sqrt{\frac{n}{\log n}} =\vcentcolon \ell_0.
		\]
		Moreover, we can find such $i$ efficiently.
		\\\\
		\noindent\textbf{Constructing $\mathsf{DS}$.}
        Using \Cref{DS}, we create a data structure $\mathsf{DS}$ 
        with $w_i$ wires and $\card{S_i} \leq k$ variables and a
        labelling consistent with that of $C$ for the wires labelled by
        variables from $S_i$.
        This can be done in time
		\[
		2^{k}\cdot \ell_0^{O(\ell_0^2)}=2^{k + O\!\left(\ell_0^2\cdot \log \ell_0\right)} \leq 2^{k+n/2}.
		\]
		
		\noindent\textbf{Enumeration.}
        For each $\rho\in\bool^{[n]\backslash S_i}$, we obtain a restricted
        circuit $\rst{C}{\rho}$ (on either $k$ or $k+1$ variables), which
        has $\ell_0$ wires (whose labelling is consistent with the one used
        for $\mathsf{DS}$ created above) and has $\poly(n)$ gates. Then
        using $\mathsf{DS}$, we can efficiently look up the number of
        satisfying assignments of $\rst{C}{\rho}$. Finally we sum over
        these numbers over all such $\rho$'s and this gives the number of
        satisfying assignments of $C$.
		
		The total running time of the above algorithm is
		\[
		2^{k + n/2} + 2^{n-k}\cdot \poly(n) = 2^{n-k}\cdot \poly(n),
		\]
		as desired.
	\end{proof}

	\section{Pseudorandom Generators and \texorpdfstring{${\sf MCSP}$}{MCSP} Lower Bounds}
	In this section, we show a PRG for small comparator circuits,
	and derive from it lower bounds for comparator circuits computing ${\sf MCSP}$.
	
	\subsection{Proof of the PRG}
	We start with some definitions and notations.
	
	\begin{itemize}
		\item For a Boolean function $f$, we denote by $\ell(f)$ the minimum number of wires in a comparator circuit computing $f$.
		\item 
            We will often describe a restriction $\rho \in \set{0,1,*}^n$ 
            as a pair
		$(\sigma, \beta) \in \blt^n \times \blt^n$.
		The string $\sigma$ is the characteristic vector of the set of
		coordinates that are assigned $*$ by $\rho$, and
		$\beta$ is an assignment of values to the remaining
		coordinates.
		The string $\sigma$ is also called a \emph{selection}.
		\item 
		We say that a distribution $\cald$ on $\set{0,1}^n$
		is a \emph{$p$-regular} random \emph{selection} if
		$\Prob_{\sigma \flws \cald} \left[ \sigma(i) = 1 \right] = p$
		for every $i \in [n]$.
		
	\end{itemize}
	
    As mentioned in \Cref{sec:techniques}, we will need a result saying
    that the number of wires in a comparator circuit shrinks with high
    probability under pseudorandom restrictions.
	
	\begin{lemma}
		\label{lemma:pseudorandom_shrinkage_wires}
		Let $c$ be a constant and let 
		$f : \blt^n \to \blt$.
		Let $\ell := \ell(f)$
		and
        $p = \ell^{-2/3}$, and
        suppose that $\ell = n^{\Omega(1)}$.
		There exists a $p$-regular pseudorandom
		selection $\mathcal{D}$ over $n$ variables that is samplable using
		$r=\polylog(\ell)$ random bits such that
		\begin{equation*}
		\Prob_{\sigma \flws \cald,\, \beta \flws \blt^n} 
		\left[ 
		\ell(\rst{f}{(\sigma, \beta)}) \geq 
		2^{3 \sqrt{c \log \ell}} \cdot p\ell
		\right]
		\leq
		2\cdot\ell^{-c}.
		\end{equation*}
		Moreover, there exists a circuit of size $\polylog(\ell)$ such that, given $j\in\bool^{\log n}$ and a seed
		$z\in\bool^{r}$, the circuit computes the $j$-th coordinate of
		$\cald(z)$.
	\end{lemma}
    The proof of \Cref{lemma:pseudorandom_shrinkage_wires} follows closely
    that of {\cite[Lemma 5.3]{DBLP:journals/jacm/ImpagliazzoMZ19}}, except
    for that here we also need to show that the pseudorandom restriction
    can be computed with small size circuits. 
    Such a restriction
    is proved to exist in Lemma~18
    of~\cite{DBLP:journals/toct/CheraghchiKLM20}. 
    For completeness, a proof is presented in
    \Cref{sec:shrinkage_app}.
	
	\begin{theorem}[Local PRGs]\label{prg}
		For every $n \in \mathbb{N}$, $\ell=n^{\Omega(1)}$, and
		$\varepsilon \geq 1/\poly(n)$, there is a pseudorandom generator
		$G\colon \bool^{r}\to\bool^{n}$, with seed length
		\[
		r=\ell^{2/3 + o(1)}
		\]
		that $\varepsilon$-fools comparator circuits with $\ell$ wires over $n$ variables. 
		Moreover, for every seed $z\in\bool^r$, there is a circuit $D_z$ of
		size $\ell^{2/3 + o(1)}$ such that, given as input $j\in[n]$, $D_z$
		computes the $j$-th bit of $G(z)$.
	\end{theorem}
	
	\begin{proof}[Proof Sketch]
		In~\cite{DBLP:journals/jacm/ImpagliazzoMZ19}, it is shown that
		if a circuit class ``shrinks'' with high probability
		under a pseudorandom restriction, then we can construct pseudorandom generators
		for this circuit class with non-trivial seed-length.
		The authors of~\cite{DBLP:journals/toct/CheraghchiKLM20} then showed that
		if the same shrinkage property holds
		for random selections that can be efficiently sampled and computed,
		then we can obtain local PRGs.
		In Lemma~\ref{lemma:pseudorandom_shrinkage_wires}, we proved 
		exactly what is required by~\cite{DBLP:journals/toct/CheraghchiKLM20}
		to obtain local PRGs for comparator circuits.
		
		More specifically, the theorem can be derived by following the proof
		of~\cite[Lemma 16]{DBLP:journals/toct/CheraghchiKLM20}, and adjusting the
		parameters there in a natural way. In particular, we will use
		$p\vcentcolon=\ell^{-2/3}$ so that after the pseudorandom restriction in
		\Cref{lemma:pseudorandom_shrinkage_wires}, the restricted comparator circuit
		has at most $\ell_0\vcentcolon=2^{O( \sqrt{\log \ell})} \cdot p\ell=2^{O(
			\sqrt{\log \ell})} \cdot \ell^{1/3}$ wires (with high probability). Another
		observation needed in the proof is that, by \Cref{lem:structural}, there can
		be at most $2^{\ell^{2/3+o(1)}}$ distinct functions for comparator circuits with
		this many wires. We omit the details here.
	\end{proof}
	
	\subsection{Proof of the \texorpdfstring{${\sf MCSP}$}{MCSP} lower bound}
	We prove the following stronger result which implies \Cref{thm:MCSP-intro}.
	\begin{theorem}
		For any $\varepsilon>0$ and any $0<\alpha \leq 1-\varepsilon$, ${\sf MCSP}[n^{\alpha}]$ on inputs of length $n$ cannot be computed by comparator circuits with $n^{1+\alpha/2 - \varepsilon}$ wires.
	\end{theorem}
	\begin{proof}
		Let $f$ denote the function ${\sf MCSP}[n^{\alpha}]$ on inputs of length $n$. For the sake of contradiction, suppose $f$ can be computed by a comparator circuit $C$ with $n^{1+\alpha/2 - \varepsilon}$ wires, for some $\varepsilon>0$.
		
		Let $k \vcentcolon= n^{\alpha + \varepsilon/2}$.
		Consider an (almost-even) partition of the $n$ variables into 
		$\floor{n/k}$
		\emph{consecutive} blocks,
		denoted as $S_1,S_2,\dots,S_{\lfloor n/k \rfloor}$. Again, by an
		averaging argument, there is some $i\in [\lfloor n/k \rfloor]$ such
		that after fixing the values of the variables outside $S_i$, the
		number of wires in the restricted circuit is at most
		\[
		\ell \vcentcolon= n^{1+\alpha/2 -\varepsilon}/\lfloor n/k \rfloor = n^{1.5\alpha -\varepsilon/2}.
		\]
		Let $\rho$ be a restriction that fixes the values of the variables outside $S_i$ to be $0$ and leaves the variables in $S_i$ unrestricted. Let $G$ be the PRG from \Cref{prg} that has seed length $r\vcentcolon=\ell^{2/3+o(1)}$ and $(1/3)$-fools comparator circuits with at most $\ell$ wires.
		
		On the one hand, since $|S_i|\geq k$, then by a counting argument, for a uniformly random $x\in\bool^{n}$, the circuit size of the truth table given by $\rho\circ x$ is at least $k/(10\log k)> n^{\alpha}$, with probability at least $1/2$. In other words,
		\[
		\Prob_{x\in\bool^{n}}[\rst{f}{\rho}(x)=1]\leq 1/2.
		\]
		On the other hand, by the second item of \Cref{prg}, for any seed $z\in\bool^r$, the output of the PRG $G(z)$, viewed as a truth table, represents a function that can be computed by a circuit of size $\ell^{2/3+o(1)}$. Then knowing $i\in [n]$ (which can be encoded using $\log(n)$ bits), the truth table given by $\rho\circ G(z)$ has circuit size at most
		\[
		\polylog(n) + \ell^{2/3+o(1)} \leq n^{\alpha}.
		\]
		This implies
		\[
		\Prob_{z\in\bool^{r}}[\rst{C}{\rho}(G(z))=1] = 1,
		\]
		which contradicts the security of $G$.
	\end{proof}
	
	\section{Learning Algorithms}\label{sec:learning}   
	Recall that a (distribution-independent) PAC learning algorithm for a class of functions $\mathcal{C}$ has access to labelled examples $(x,f(x))$ from an unknown function $f \in \mathcal{C}$, where $x$ is sampled according to some (also unknown) distribution $\mathcal{D}$. The goal of the learner is to output, with high probability over its internal randomness and over the choice of random examples, a hypothesis $h$ that is close to $f$ under $\mathcal{D}$. As in \cite{DBLP:conf/innovations/ServedioT17}, here we consider the stronger model of ``randomized exact learning from membership and equivalence queries''. It is known that learnability in this model implies learnability in the distribution-independent
	PAC model with membership queries (see {\cite[Section 2]{DBLP:conf/innovations/ServedioT17}} and the references therein).

	\begin{theorem}[{\cite[Lemma 4.4]{DBLP:conf/innovations/ServedioT17}}]\label{thm:ST}
		Fix any partition $S_1,S_2,\dots,S_{n^{1-n^{\delta}}}$ of $[n]$ into equal-size subsets, where each $S_i$ is of size $n^{\delta}$ and $\delta>0$. Let $\mathcal{C}$ be a class of $n$-variate functions such that for each $f\in \mathcal{C}$, there is an $S_i$ such that $\left|\left\{\rst{f}{\rho}\right\}_{\rho\in\bool^{[n]\backslash S_i}}\right|\leq 2^{n^{\beta}}$, where $\beta<1$ and moreover $\delta+\beta<1$. Then there is a randomized  exact  learning algorithms for $\mathcal{C}$ that uses membership and equivalence queries and runs in time $2^{n-n^{\delta}}\cdot\poly(n)$.
	\end{theorem}
	
	\begin{corollary}
		For every $\varepsilon>0$, there is a randomized exact learning algorithms for comparator circuits with $n^{1.5-\varepsilon}$ wires that uses membership and equivalence queries that runs in time $2^{n-n^{\Omega(\varepsilon)}}\cdot\poly(n)$.
	\end{corollary}
	\begin{proof}
		Consider \Cref{thm:ST} and any partition $S_1,S_2,\dots,S_{n^{1-n^{\delta}}}$ of the $n$ variables into equal-size subsets, each is of size $n^{\delta}$, where $\delta\vcentcolon=\varepsilon/3$. Then by an averaging argument, for every comparator circuit $C$ with $n^{1.5-\varepsilon}$ wires, there is some $S_i$ such that after fixing the variables outside of $S_i$, the number of wires in the restricted circuit is at most $\ell\vcentcolon=n^{1.5-\varepsilon}/n^{1-\delta}\leq n^{.5-2\varepsilon/3}$. By \Cref{lem:structural}, such a restricted circuit computes some function that is equivalent to a circuit with $\ell(\ell-1)/2$ gates, and there are at most $\ell^{O\!\left(\ell^2\right))} \leq 2^{n^{1-\varepsilon/2}}$
		such circuits. Therefore we have
		\[
		\left|\left\{\rst{C}{\rho}\right\}_{\rho\in\bool^{[n]\backslash S_i}}\right|\leq 2^{n^{\beta}},
		\]
		where $\beta \vcentcolon= 1-\varepsilon/2<1$ and $\delta+\beta<1$. The algorithm then follows from \Cref{thm:ST}.
	\end{proof}
	
    \section*{Acknowledgements}

    B. P. Cavalar acknowledges support of the Chancellor's
    International Scholarship of the University of Warwick.
    Z. Lu acknowledges support from the Royal Society University
    Research Fellowship URF\textbackslash R1\textbackslash191059.
    Both authors are indebted to Igor C. Oliveira
    for numerous helpful 
    discussions and comments.

	\bibliographystyle{alpha}   
	
	\bibliography{refs}

	\appendix
	
	\section{Proof of \texorpdfstring{\Cref{thm:general}}{}}\label{sec:generalization-proof}
	
    \noindent\textbf{The hard function.} We need to slightly modify the
    hard function in \Cref{def-And} (particularly the function $\alpha$) to
    adjust an arbitrary partition as in \Cref{thm:general}. For an integer
    $k$ and a partition of $n$ variables into $n/k$ equal-sized blocks,
    denoted by 
    $S\vcentcolon = \set{S_1,S_2,\dots,S_{n/k}}$,
    define $A_{S,k}\colon\bool^{n+n}\to\bool$ as follows:
	\[
	A_{S,k}(x_1,\dots,x_{n},y_1,\dots,y_{n}) \vcentcolon= \mathrm{Enc}(x_1,\dots,x_{n})_{\alpha(y_1,\dots,y_{n})},
	\]
	where $\mathrm{Enc}$ is the code from \Cref{ecc} that maps $n$ bits to $2^{k}$ bits, and $\alpha\colon\bool^n\to\bool^k$ is defined as
	\[
	\alpha(y_1,\dots,y_{n}) \vcentcolon= \left(\bigoplus_{z\in B_1} z, \bigoplus_{z\in B_2} z, \dots, \bigoplus_{z\in B_k} z\right),
	\]
	where
	$B_j \vcentcolon= \bigcup_{i\in[n/k]} \left\{z\colon z \text{ is the $j$-th variables of } S_i\right\}$.
	
	\vspace{0.2cm}
	\noindent\textbf{Good $x$.} We will need the following lemma which says that for most $x\in\bool^n$, the codeword of $x$ is hard to approximate for any fixed small set of functions. 
	\begin{lemma}\label{lem:counting}
        Let $k$ be such that
        $c\cdot\log n\leq k \leq n/3$,
        where $c$ is the constant from
        \Cref{ecc}, and let	$\mathrm{Enc}$ be the code from \Cref{ecc} that
        maps $n$ bits to $2^{k}$ bits. Let $\mathcal{H}'$ be a set of
        $k$-input-bit Boolean functions such that
        $\left|\mathcal{H}'\right|\leq 2^{2n/3}$. Then, 
        with probability
        at least $1-1/2^{n/2}$
        over a random $x\in\bool^n$, 
        the following holds
        for every
        $f\in\mathcal{H}'$:
		\begin{equation}\label{eq:counting}
		\Prob_{z\in\bool^k}\left[f(z)=\mathrm{Enc}(x)_z\right]\leq \frac{1}{2} + \frac{n}{2^{k/4}}.
		\end{equation}
	\end{lemma}
	\begin{proof}
        The proof is by a counting argument. For every $f\in\mathcal{H}'$,
        consider the $2^{k}$-bit string $\mathrm{tt}(f)$ which is the truth
        table computed by $f$. Let us say $x$ is bad for $f$ if
        \Cref{eq:counting} 
        does not hold,
        \bnote{I changed ``holds'' to ``does not hold''}
        \znote{looks good. thanks!}
        which means that $\mathrm{tt}(f)$ and
        $\mathrm{Enc}(x)$ agree on more than $1/2+n/2^{k/4}$ positions. By
        the list-decodability of $\mathrm{Enc}$, the number of such $x$'s
        is at most $O\!\left(2^{k/2}/n\right)$. By an union bound over all
        the $2^{2n/3}$ 
        functions
        in $\mathcal{H}'$, the fraction of bad $x$'s
        is at most
		\[
		\frac{O\!\left(2^{k/2}/n\right) \cdot  2^{2n/3}}{2^n} < \frac{1}{2^{n/2}}, 
		\]
		as desired.
	\end{proof}
	
	We are now ready to prove \Cref{thm:general}.
	
	\begin{proof}[Proof of \Cref{thm:general}]
		Let $A \vcentcolon= A_{S,k}$ be the hard function on $2n$ variables defined as above, where $S$ is the partition in the statement of the theorem, and let 
		\[
		B_j \vcentcolon= \bigcup_{i\in[n/k]} \left\{z\colon z \text{ is the $j$-th variables of } S_i\right\}.
		\]
		Also, let $\mathcal{H}'$ be the set of $k$-input-bit Boolean
        functions defined as follows:
		\[
		\mathcal{H}' \vcentcolon= \left\{f \colon \text{$\exists h \in \mathcal{H}$ and $w\in\bool^k$, such that $f(z) =h(z\oplus w)$ for all $z\in\bool^k$} \right\}.
		\]
        That is, $\mathcal{H}'$ is the set of all possible ``shifted''
        functions in $\mathcal{H}$. By \Cref{lem:counting}, with
        probability 
        at least $1-1/2^{n/2}$ 
        over a random $x\in\bool^n$, 
        for every $f\in\mathcal{H}'$
        we have
		\begin{equation}\label{eq:counting-2}
		\Prob_{z\in\bool^k}\left[f(z)=\mathrm{Enc}(x)_z\right]\leq \frac{1}{2} + \frac{n}{2^{k/4}}.
		\end{equation}
		Let us call $x$ \emph{good} if it satisfies \Cref{eq:counting-2}.
		
        To show the theorem, we need to upper bound the following
        probability, for every circuit $C_0\in \mathcal{C}_{2n}$ of size
        $s(n,k)$:
		\begin{align*}
		\Prob_{x,y\in\bool^{n}\times \bool^{n}}[A(x,y)=C_0(x,y)] &\leq \Prob_{x,y}[A(x,y)=C_0(x,y) \mid \text{$x$ is good}]  +\Prob_{x}[\text{$x$ is not good}]\\
		&\leq \Prob_{x,y}[A(x,y)=C_0(x,y) \mid \text{$x$ is good}]+ \frac{1}{2^{n/2}}.
		\end{align*}
		Let $x$ be any fixed $n$-bit string that is good. Let $A'\colon\bool^n\to\bool$ be
		\[
		A'(y) \vcentcolon=A(x,y), 
		\]
		and let $C$ be the circuit defined as
		\[
		C(y) \vcentcolon=C_0(x,y).
		\]
		Note that since the class $\mathcal{C}$ is closed under restriction, $C$ is a circuit from $\mathcal{C}_{n}$ with size at most $s(n,k)$.
		We will show that
		\[
		\Prob_{y\in\bool^{n}}\left[A'(y)=C(y)\right]\leq \frac{1}{2}+ \frac{n}{2^{k/4}}.
		\]
		Let $S_i$ be the block in the assumption of the theorem such that 
		\[
		\left\{\rst{C}{\rho}\right\}_{\rho\in\bool^{[n]\backslash S_i}}\subseteq\mathcal{H}.
		\]
		We have
		\[
		\Prob_{y\in\bool^{n}}\left[A'(y)=C(y)\right] = \Prob_{\rho\in\bool^{[n]\backslash S_i},z\in\bool^{k}}\left[\rst{A'}{\rho}(z)=\rst{C}{\rho}(z)\right].
		\]
		It suffices to upper bound
		\[
		\Prob_{z\in\bool^{k}}\left[\rst{A'}{\rho}(z)=\rst{C}{\rho}(z)\right]
		\]
		for every $\rho\in\bool^{[n]\backslash S_i}$.
		For the sake of contradiction, suppose for some $\rho$, we have
		\begin{equation}\label{ext-property-2}
		\frac{1}{2} +\frac{n}{2^{k/4}} < \Prob_{z\in\bool^{k}}\left[\rst{A'}{\rho}(z)=\rst{C}{\rho}(z)\right]
		= \Prob_{z\in\bool^{k}}\left[\mathrm{Enc}(x)_{\alpha}=\rst{C}{\rho}(z)\right],
		\end{equation}
		where $\alpha\in\bool^k$ is 
		\[
		\alpha_j \vcentcolon= \parity\!\left(\rho|_{B_j \backslash S_i}\right) \oplus z_j,
		\]
		and $\rho|_{B_j \backslash S_i}$ denotes the partial assignment given by $\rho$ but restricted to only variables in the set $B_j \backslash S_i$. That is, $\alpha$ is some ``shift'' of $z$, so $\alpha$ is uniformly distributed for uniformly random $z$.
		Therefore, \Cref{ext-property-2} implies
		\[
		\Prob_{z\in\bool^{k}}\left[\mathrm{Enc}(x)_{z}=\rst{C}{\rho}(z\oplus w)\right] > \frac{1}{2} +\frac{n}{2^{k/4}},
		\]
        for some $w\in\bool^k$. This gives a function in $\mathcal{H}'$
        that computes $\mathrm{Enc}(x)$ on more than $1/2 +n/2^{k/4}$
        positions, which contradicts the assumption that $x$ is good.
	\end{proof}
	
	\section{Pseudorandom Shrinkage for Comparator Circuits: Proof of \texorpdfstring{\Cref{lemma:pseudorandom_shrinkage_wires}}{Lemma 17}}\label{sec:pseudo_shrinkage}


	\vspace{0.2cm}
    \noindent\textbf{Technical tools.} 
	We will need a Chernoff-Hoeffding bounds for distributions with
	bounded independence from~\cite{DBLP:journals/siamdm/SchmidtSS95}
	(Lemmas 2.3 in~\cite{DBLP:journals/jacm/ImpagliazzoMZ19}). Recall  that a distribution $\cald$ 
	on $[m]^n$ 
	is \emph{$k$-wise independent} if, for any set
	$A \sseq [n]$ of size $\card{A} \leq k$, 
	the random variables
	$\set{\sigma(i) : i \in A}$ are mutually independent when
	$\sigma \flws \cald$.
	
	\begin{lemma}[\cite{DBLP:journals/siamdm/SchmidtSS95}]
		\label{lemma:general_chernoff_bounded}
		Let $a_1,\dots, a_n \in \bbr_+$ and let $m = \max_i a_i$.
		Suppose that $X_1,\dots,X_n \in \blt$ are $k$-wise independent
		random variables with $\Prob[X_i=1]=p$.
		Let $X = \sum_{i}a_i X_i$ and $\mu = \Exp[X] = p \sum_{i} a_i$.
		We have $\Prob[X \geq 2k(m+\mu)] \leq 2^{-k}$.
	\end{lemma}
	
	\begin{lemma}[{\cite[Lemma 2.4]{DBLP:journals/jacm/ImpagliazzoMZ19}}]
		\label{lemma:small_chernoff_bounded}
		Let $X_1,\dots,X_n \in \blt$ be $k$-wise independent random
		variables with $\Prob[X_i=1]=p$.
		Let $X = \sum_{i} X_i$ and $\mu  = \Exp[X] = np$.
		We have $\Pr[X \geq k] \leq \mu^k / k!$.
	\end{lemma}

	\vspace{0.2cm}
    \noindent\textbf{Shrinkage of comparator circuits under pseudorandom
    restrictions.}\label{sec:shrinkage_app}
	We first show the following result for comparator circuits which is analogous to \cite[Lemma 5.2]{DBLP:journals/jacm/ImpagliazzoMZ19} for branching programs.
	\begin{lemma}
		\label{lemma:recursive-cc}
		Let $f : \blt^n \to \blt$ be a Boolean function,
		and let $H \sseq [n]$.
		For $h \in \blt^H$, let $\rho_h$ denote the restriction that
		sets the variables in $H$ to $h$, and leaves the other variables
		free.
		We have
		$\ell(f) \leq 2^{\card{H}} \cdot
		\left( \max_{h \in \blt^H} \ell(\rst{f}{\rho_h}) + \card{H} \right)$.
	\end{lemma}
	\begin{proof}
		For $h \in \blt^H$,
		let 
		$\one_{h} : x \mapsto \one\set{x = h}$.
		Clearly, $\one_h$ can be computed by a comparator
		circuit with $\card{H}$ wires.
		Since
		$f = \bigvee_{h \in \blt^H} (\one_h \land \rst{f}{\rho_h})$,
		the result follows.
	\end{proof}
	
	\begin{lemma}[Reminder of \Cref{lemma:pseudorandom_shrinkage_wires}]
		\label{lemma:pseudorandom_shrinkage_wires_app}
		Let $c$ be a constant and let 
		$f : \blt^n \to \blt$.
		Let $\ell := \ell(f)$
		and
        $p = \ell^{-2/3}$, and
        suppose that $\ell = n^{\Omega(1)}$.
		There exists a $p$-regular pseudorandom
		selection $\mathcal{D}$ over $n$ variables that is samplable using
		$r=\polylog(\ell)$ random bits such that
		\begin{equation*}
		\Prob_{\sigma \flws \cald,\, \beta \flws \blt^n} 
		\left[ 
		\ell(\rst{f}{(\sigma, \beta)}) \geq 
		2^{3 \sqrt{c \log \ell}} \cdot p\ell
		\right]
		\leq
		2\cdot\ell^{-c}.
		\end{equation*}
		Moreover, there exists a circuit of size $\polylog(\ell)$ such that, given $j\in\bool^{\log n}$ and a seed
		$z\in\bool^{r}$, the circuit computes the $j$-th coordinate of
		$\cald(z)$.
	\end{lemma}
	\begin{proof}
		First, we note that a $k$-wise independent random selection
		that can be efficiently sampled and computed with the required
		parameters
		is proved to exist in Lemma~18
		of~\cite{DBLP:journals/toct/CheraghchiKLM20}.
		Henceforth, we let $\rho$ be the random restriction
		described by the pair $(\sigma, \beta)$.
		
		Let $C$ be a comparator circuit with $\ell$ wires computing $f$.
		Let $k = c \cdot \log \ell$.
		For $i \in [n]$, let $w_i$ be the number of wires in $C$ labelled
		with the variable $x_i$.
		
		Let $\alpha = \sqrt{c / \log \ell}$.
		We say that $i \in [n]$ is \emph{heavy} if
		$w_i \geq p^{1-\alpha}\cdot \ell$ and \emph{light} otherwise.
		Let $H \sseq [n]$ be the set of heavy variables.
		We have $\card{H} \leq (1/p)^{1-\alpha}$.
		Let also $H(\rho) := H \cap \rho^{-1}(*)$.
		Let
		$\rho'$ be a restriction such that
		$\rho'(x) = \rho(x)$ for $x \notin H(\rho)$
		and which sets the variables in $H(\rho)$ so as to maximize
		$\ell(\rst{f}{\rho'})$.
		By Lemma~\ref{lemma:recursive-cc}, we have
		$\ell(\rst{f}{\rho}) \leq 2^{\card{H(\rho)}+1} \cdot
		\ell(\rst{f}{\rho'})$.
		
		We now let $h = \ceil{3/2 \cdot c/\alpha}$,
		and observe that
		\begin{equation*}
		\Prob_{\rho}
		\left[  
		\ell(\rst{f}{\rho}) \geq 2^{h+3} k p^{1-\alpha}s
		\right]
		\leq
		\Prob_{\rho}
		\left[  
		\card{H(\rho)} \geq h
		\right]
		+
		\Prob_{\rho}
		\left[  
		\ell(\rst{f}{\rho'}) \geq 4kp^{1-\alpha} \ell
		\right].
		\end{equation*}
		Let $X_i$ be a random variable such that
		$X_i = 1$ iff $\rho(i)=*$.
		From Lemma~\ref{lemma:small_chernoff_bounded}, it follows that the
		first term can be bounded by $(\card{H}p)^{h} \leq p^{\alpha h} \leq \ell^{-c}$.
		For the second term, we can apply
		Lemma~\ref{lemma:general_chernoff_bounded} 
		on the light variables
		with $\mu \leq p \ell$
		and $m < p^{1-\alpha} \ell$, so that $m+\mu \leq 2p^{1-\alpha} \ell$,
		thus bounding the probability by $2^{-k} \leq \ell^{-c}$.
	\end{proof}

\end{document}